\documentclass[conference,a4paper]{IEEEtran}

\usepackage[utf8]{inputenc}
\usepackage[T1]{fontenc}
\usepackage{url}
\usepackage{ifthen}
\usepackage{cite}
\usepackage[cmex10]{amsmath}

\usepackage{graphicx}
\usepackage{booktabs}
\usepackage{amsfonts}
\usepackage{amsmath}
\usepackage{amssymb}
\usepackage{mathrsfs}
\usepackage{color}
\usepackage{bm}
\usepackage{bbm}
\usepackage{latexsym}

\newtheorem{thm}{Theorem}

\newtheorem{lem}{Lemma}
\newtheorem{proof}{proof}

\newtheorem{defn}{Definition}
\newtheorem{rem}{Remark}

\begin{document}

\title{New Construction of $q$-ary Codes Correcting a Burst of
at most $t$ Deletions}


\author{\IEEEauthorblockN{Wentu Song\IEEEauthorrefmark{1},
Kui Cai\IEEEauthorrefmark{1} and Tony Q. S.
Quek\IEEEauthorrefmark{2}} \IEEEauthorblockA{
    \IEEEauthorrefmark{1}Science, Mathematics and Technology
    Cluster, Singapore University of Technology and Design,
    Singapore 487372\\
    \IEEEauthorrefmark{2}Information Systems Technology and Design
    Pillar, Singapore University of Technology and Design,
    Singapore 487372\\ Email: \{wentu\_song,
    cai\_kui, tonyquek\}@sutd.edu.sg}}

\maketitle

{\let\thefootnote\relax\footnotetext{This work was supported by
SUTD Kickstarter Initiative (SKI) Grant  2021\_04\_05 and the
Singapore Ministry of Education Academic Research Fund Tier 2
T2EP50221-0036.}}

\begin{abstract}
In this paper, for any fixed positive integers $t$ and $q>2$, we
construct $q$-ary codes correcting a burst of at most $t$
deletions with redundancy $\log n+8\log\log n+o(\log\log
n)+\gamma_{q,t}$ bits and near-linear encoding/decoding
complexity, where $n$ is the message length and $\gamma_{q,t}$ is
a constant that only depends on $q$ and $t$. In previous works
there are constructions of such codes with redundancy $\log
n+O(\log q\log\log n)$ bits or $\log n+O(t^2\log\log n)+O(t\log
q)$. The redundancy of our new construction is independent of $q$
and $t$ in the second term.
\end{abstract}

\section{Introduction}
Study of deletion/insertion correcting codes, which was originated
in 1960s, has made a great progress in recent years. One of the
basic problem is to construct codes with low redundancy and low
encoding/decoding complexity, where the redundancy of a $q$-ary
$(q\geq 2)$ code $\mathcal C$ of length $n$ is defined as
$n-\log_q|\mathcal C|$ in symbol or $(n-\log_q|\mathcal C|)\log q$
in bits.\footnote{In this paper, for any real $x>0$, for
simplicity, we write $\log_2x=\log x$.}

The famous VT codes were proved to be a family of single-deletion
correcting binary codes and are asymptotically optimal in
redundancy \cite{Levenshtein65}. The VT construction was
generalized to nonbinary single-deletion correcting codes in
\cite{Tenengolts84}, and to a new version in \cite{Nguyen-23}
using differential vector, with asymptotically optimal redundancy
and efficient encoding/decoding. Other works in binary and
nonbinary codes for correcting multiple deletions can be found in
\cite{Brakensiek18}-\cite{S-Liu-23} and the references therein.

Burst deletions and insertions, which means that deletions and
insertions occur at consecutive positions in a string, are a class
of errors that can be found in many applications, such as
DNA-based data storage and file synchronization. For binary case,
the maximal cardinality of a $t$-burst-deletion correcting code
$($i.e., a code that can correct a burst of \emph{exactly} $t$
deletions$)$ is proved to be asymptotically upper bounded by
$2^{n-t+1}/n$ \cite{Levenshtein67}, so its redundancy is
asymptotically lower bounded by $\log n+t-1$. Several
constructions of binary codes correcting a burst of exactly $t$
deletions have been reported in \cite{L-Cheng-14,Schoeny2017},
where the construction in \cite{Schoeny2017} achieves an optimal
redundancy of $\log n+(t-1)\log\log n+k-\log k$. A more general
class, i.e., codes correcting a burst of \emph{at most} $t$
deletions, were also constructed in the same paper
\cite{Schoeny2017}, and this construction was improved in
\cite{R-Gabrys-17} to achieve a redundancy of $\lceil\log
t\rceil\log n+(t(t+1)/2-1)\log\log n+c_t$ for some constant $c_t$
that only depends on $t$. In \cite{Lenz20}, by using VT constraint
and shifted VT constraint in the so-called $(\bm p,\delta)$-dense
strings, binary codes correcting a burst of at most $t$ deletions
were constructed, with an optimal redundancy of $\log
n+t(t+1)/2\log\log n+c'_t$, where $c'_t$ is a constant depending
only on $t$.

In the recent parallel works \cite{Wang-S-22} and \cite{Wentu23},
$q$-ary codes correcting a burst of at most $t$ deletions were
constructed for even integer $q>2$, with redundancy $\log n+O(\log
q\log\log n)$, or more specifically, $\log n+(8\log q+9)\log\log
n+\gamma_t'+o(\log\log n)$ bits for some constant $\gamma_t'$ that
only depends on $t$. The basic techniques in \cite{Wang-S-22} and
\cite{Wentu23} are to represent each $q$-ary string as a binary
matrix whose column are the binary representation of the entries
of the corresponding $q$-ary string, with the constraint that the
first row of the matrix representation is $(\bm p,\delta)$-dense.
Then the first row of the matrix is protected by binary burst
deletion correcting codes of length $n$ and the other rows are
protected by binary burst deletion correcting codes of length not
greater than $2\delta$, which results in the redundancy of $\log
n+O(\log q\log\log n)$ bits of the constructed code. A different
construction of $q$-ary codes correcting a burst of at most $t$
deletions was reported in a more recent work \cite{Nguyen-23},
which has redundancy $\log n+O(t^2\log\log n)+O(t\log q)$.

In this paper, we construct $q$-ary codes correcting a burst of at
most $t$ deletions for any fixed $t$ and $q>2$. We consider
$q$-ary $(\bm p,\delta)$-dense strings, which are defined similar
to binary $(\bm p,\delta)$-dense strings as in \cite{Lenz20}, and
give an efficient algorithm for encoding and decoding of $q$-ary
$(\bm p,\delta)$-dense strings. In our construction, a VT-like
function is used to locate the deletions within an interval of
length not greater than $3\delta$, which results in $\log n$ bits
in redundancy. In addition, two functions are used to recover the
substring destroyed by deletions, which results in $8\log\log
n+o(\log\log n)+\gamma_{q,t}$ bits in redundancy, where
$\gamma_{q,t}$ is a constant that only depends on $q$ and $t$.
Thus, the total redundancy of our construction is $\log
n+8\log\log n+o(\log\log n)+\gamma_{q,t}$ bits. The
encoding/decoding complexity of our construction is
$O(q^{7t}n(\log n)^3)$. Compared to previous work, the redundancy
of our new construction is independent of $q$ and $t$ in the
second term.

In Section \uppercase\expandafter{\romannumeral 2}, we introduce
related definitions and notations. In Section
\uppercase\expandafter{\romannumeral 3}, we study pattern dense
$q$-ary strings. Our new construction of $q$-ary burst-deletion
correcting codes is given in Section
\uppercase\expandafter{\romannumeral 4}, and the paper is
concluded in Section \uppercase\expandafter{\romannumeral 5}.

\section{Preliminaries}

Let $[m,n]=\{m,m+1,\ldots,n\}$ for any two integers $m$ and $n$
such that $m\leq n$ and call $[m,n]$ an \emph{interval}. If $m>n$,
then let $[m,n]=\emptyset$. For simplicity, we denote $[n]=[1,n]$
for any positive integer $n$. The size of a set $S$ is denoted by
$|S|$.

Given any integer $q\geq 2$, let $\Sigma_q=\{0,1,2,\cdots,q-1\}$.
For any sequence (also called a string or a vector)
$\bm{x}\in\Sigma_q^n$, $n$ is called the length of $\bm{x}$ and
denote $|\bm x|=n$. We will denote $\bm{x}=(x_1,x_2,\ldots,x_n)$
or $\bm{x}=x_1x_2\cdots x_n$. For any set $I=\{i_1, i_2, \ldots,
i_m\}\subseteq [n]$ such that $i_1<i_2<\cdots<i_m$, denote
$x_I=x_{i_1} x_{i_2} \cdots x_{i_m}$ and call $x_I$ a
\emph{subsequence} of $\bm x$. If $I=[i,j]$ for some $i,j\in[1,n]$
such that $i\leq j$, then $x_{I}=x_{[i,j]}=x_{i}x_{i+1}\cdots
x_{j}$ is called a \emph{substring} of $\bm x$. We say that $\bm
x$ \emph{contains} $\bm p~($or $\bm p$ is contained in $\bm x)$ if
$\bm{p}$ is a substring of $\bm{x}$. For any $\bm x\in\Sigma_q^n$
and $\bm y\in\Sigma_q^{n'}$, we use $\bm x\bm y$ to denote their
\emph{concatenation}, i.e., $\bm x\bm y=x_1x_2\cdots
x_{n}y_1y_2\cdots y_{n'}$. We also use notations such as $\bm x_0,
\bm x_1, \cdots, \bm x_k$ to denote substrings of a sequence $\bm
x$. For example, the notation $\bm x=\bm x_1\bm x_2\cdots\bm x_k$
means that the sequence $\bm x$ consists of $k$ substrings $\bm
x_1, \bm x_2, \cdots, \bm x_k$.

Let $t\leq n$ be a nonnegative integer. For any
$\bm{x}\in\Sigma_q^n$, let $\mathcal D_t(\bm x)$ denote the set of
subsequences of $\bm x$ of length $n-t$, and let $\mathcal
B_{t}(\bm x)$ denote the set of subsequences $\bm y$ of $\bm x$
that can be obtained from $\bm x$ by a burst of $t$ deletions,
that is $\bm y=x_{[n]\backslash D}$ for some interval
$D\subseteq[n]$ of length $t~($i.e., $D=[i,i+t-1]$ for some
$i\in[n-t+1])$. Moreover, let $\mathcal B_{\leq t}(\bm
x)=\bigcup_{t'=0}^t\mathcal B_{t'}(\bm x)$, i.e., $\mathcal
B_{\leq t}(\bm x)$ is the set of subsequences of $\bm x$ that can
be obtained from $\bm x$ by a burst of at most $t$ deletions.
Clearly, $\mathcal B_{ 1}(\bm x)=\mathcal D_{1}(\bm x)$ and
$\mathcal B_{t}(\bm x)\subseteq\mathcal D_{t}(\bm x)$ for $t\geq
2$.

A code $\mathcal C\subseteq\Sigma_q^n$ is said to be a
$t$-\emph{deletion correcting code} if for any codeword $\bm
x\in\mathcal C$, given any $\bm y\in\mathcal D_t(\bm x)$, $\bm x$
can be uniquely recovered from $\bm y$; the code $\mathcal
C\subseteq\Sigma_q^n$ is said to be capable of \emph{correcting a
burst of at most $t$ deletions} if for any $\bm x\in\mathcal C$,
given any $\bm y\in\mathcal B_{\leq t}(\bm x)$, $\bm x$ can be
uniquely recovered from $\bm y$. In this paper, we will always
assume that $q$ and $t$ are constant with respect to $n$.

Let $\bm c=(c_1,c_2,\cdots,c_n)\in\Sigma_2^n$. The VT syndrome of
$\bm c$ is defined as
$$\mathsf{VT}(\bm c)=\sum_{i=1}^nic_i~\mod (n+1).$$
It was proved in \cite{Levenshtein65} that for any $\bm
c\in\Sigma_2^n$, given $\mathsf{VT}(\bm c)$ and any $\bm
y\in\mathcal D_1(\bm c)$, one can uniquely recover $\bm x$.

If $q>2$, for each $\bm x=(x_1,x_2,\cdots,x_n)\in\Sigma_q^n$, let
$\phi(\bm x)=(\phi(\bm x)_1,\phi(\bm x)_2,\cdots,\phi(\bm
x)_{n})\in\Sigma_2^{n}$ such that $\phi(\bm x)_1=0$ and for each
$i\in[2,n]$, $\phi(\bm x)_i=1$ if $x_{i}\geq x_{i-1}$ and
$\phi(\bm x)_i=0$ if $x_{i}<x_{i-1}$. $($One can also let
$\phi(\bm x)_1=1$ for all $\bm x\in\Sigma_q^n.)$ Then we have
$q$-ary codes for correcting a single deletion.

\begin{lem}\cite{Tenengolts84}\label{q-sgl-del}
For any $\bm x\in\Sigma_q^n$, given $\mathsf{VT}(\phi(\bm
x)_{[2,n]})$, $\mathsf{Sum}(\bm x)$ and any $\bm y\in\mathcal
D_1(\bm x)$, one can uniquely recover $\bm x$, where $\phi(\bm
x)_{[2,n]}=(\phi(\bm x)_2,\cdots,\phi(\bm x)_{n})$ and
$$\mathsf{Sum}(\bm x)=\sum_{i=1}^nx_i\mod q.$$
\end{lem}

The following lemma generalizes the construction in
\cite{Sima20-1} to $q$-ary codes $(q>2)$ and will be used in our
new construction.
\begin{lem}\label{lem-sgbd-a}
Suppose that $q$ and $t$ are constants with respect to $n$. There
exists a function $h:\Sigma_q^n\rightarrow\Sigma_q^{4\log_q
n+o(\log_q n)}$, computable in time $O(q^tn^3)$, such that for any
$\bm x\in\Sigma_q^n$, given $h(\bm x)$ and any $\bm y\in\mathcal
B_{\leq t}(\bm x)$, one can uniquely recover $\bm x$.
\end{lem}
\begin{proof}
The function $h$ can be constructed by the syndrome compression
technique developed in Section II of \cite{Sima20-1}.

For each $\bm x\in\Sigma_q^n$, let $\mathcal N_t(\bm x)$ be the
set of all $\bm x'\in\Sigma_q^n\backslash\{\bm x\}$ such that
$\mathcal B_{\leq t}(\bm x)\cap\mathcal B_{\leq t}(\bm
x')\neq\emptyset$. By simple counting, we have
\begin{align}\label{num-nghb}|\mathcal N_t(\bm x)|\leq
tn^2q^{t}.\end{align}

We first construct a function $\bar
h:\Sigma_q^n\rightarrow[0,2^{\overline{R}}-1]$ such that 1)
$\overline{R}=\frac{t(t+1)}{2}\left(\log(n+1)+\log q\right)$; and
2) $\bar h(\bm x)\neq \bar h(\bm x')$ for all $\bm x\in\Sigma_q^n$
and $\bm x'\in\mathcal N_t(\bm x)$. Specifically, $\bar h$ is
constructed as follows: For each $t'\in[t]$ and $j\in[t']$, let
$$\bar h_{t',j}(\bm x)
=\big(\mathsf{VT}(\phi(x_{I_{t',j}})_{[2,n_{t',j}]}),
\mathsf{Sum}(x_{I_{t',j}})\big),$$ where $I_{t',j}=\{\ell\in[n]:
\ell\equiv j\mod t'\}$ and $n_{t',j}=|I_{t',j}|$. Then let
$$\bar h=\big(\bar h_{1,1},\bar h_{2,1},\bar h_{2,2},\cdots,\bar
h_{t,1},h_{t,2},\cdots,\bar h_{t,t}\big)$$ and view $\bar{h}(\bm
x)$ as the binary representation of a nonnegative integer.
Clearly, $|I_{t',j}|\leq\lceil\frac{n}{t'}\rceil$ and so the
length $|\bar{h}(\bm x)|$ of $\bar{h}(\bm x)$ satisfies
\begin{align*}|\bar{h}(\bm
x)|&=\log\left(\prod_{t'=1}^t\prod_{j=1}^{t'}
q\left(n_{t',j}+1\right)\right)\\
&\leq\log\left(\prod_{t'=1}^t
\left(\left\lceil\frac{n}{t'}\right\rceil+1\right)^{t'}q^{t'}\right)\\
&\leq\frac{t(t+1)}{2}\left(\log(n+1)+\log
q\right)\\&=\overline{R}.\end{align*} Hence, we have $\bar{h}(\bm
x)\in[0,2^{\overline{R}}-1]$ for all $\bm x\in\Sigma_q^n$.
Moreover, for each $t'\in[t]$, if $\bm y\in\mathcal B_{t'}(\bm
x)$, then we have $y_{I'_{t',j}}\in\mathcal D_1(x_{I_{t',j}})$ for
each $j\in[t']$, where $I'_{t',j}=\{\ell\in[n-t']: \ell\equiv
j~(\text{mod}~t')\}$. By Lemma \ref{q-sgl-del}, $x_{I_{t',j}}$ can
be recovered from $\bar h_{t',j}(\bm x)$ and $y_{I'_{t',j}}$, and
so $\bm x$ can be recovered from $\bm y$ and $\bar h(\bm x)$.
Equivalently, if $\bm x'\in\mathcal N_t(\bm x)$, then $\bar h(\bm
x)\neq \bar h(\bm x')$.

For each $\bm x\in\Sigma_q^n$, let $\mathcal P(\bm x)$ be the set
of all positive integers $j$ such that $j$ is a divisor of $|\bar
h(\bm x)-\bar h(\bm x')|$ for some $\bm x'\in\mathcal N_t(\bm x)$.
By the same discussions as in the proof of \cite[Lemma
4]{Sima20-1}, we can obtain $|\mathcal P(\bm x)|\leq
2^{\log|\mathcal N_t(\bm x)|+o(\log n)}\leq O(q^tn^3)$. $($Note
that $q$ and $t$ are assumed to be constant with respect to $n$
and, by \eqref{num-nghb}, $|\mathcal N_t(\bm x)|\leq tn^2q^{t}.)$
So, by brute force search, one can find, in time $2^{\log|\mathcal
N_t(\bm x)|+o(\log n)}\leq O(q^tn^3)$, a positive integer
$\alpha(\bm x)\leq 2^{\log|\mathcal N_t(\bm x)|+o(\log n)}$ such
that $\alpha(\bm x)\notin\mathcal P(\bm x)$. Let $h(\bm
x)=(\alpha(\bm x),\bar h(\bm x)\!\mod \alpha(\bm x))$. Then we
have $h(\bm x)\neq h(\bm x')$ for all $\bm x'\in\mathcal N_t(\bm
x)$. Equivalently, given $h(\bm x)$ and any $\bm y\in\mathcal
B_{\leq t}(\bm x)$, one can uniquely recover $\bm x$.

Moreover, since $\alpha(\bm x)\leq 2^{\log|\mathcal N_t(\bm
x)|+o(\log n)}$ is a positive integer and by \eqref{num-nghb},
$|\mathcal N_t(\bm x)|\leq tn^2q^{t}$, so viewed as a $q$-ary
sequence, we have $h(\bm x)\in\Sigma_q^{4\log_q n+o(\log_q n)}$,
which completes the proof.
\end{proof}

\section{Pattern Dense Sequences}

The concept of $(\bm p,\delta)$-dense sequences was introduced in
\cite{Lenz20} and was used to construct binary codes with
redundancy $\log n+\frac{t(t+1)}{2}\log\log n+c_t$ for correcting
a burst of at most $t$ deletions, where $n$ is the message length
and $c_t$ is a constant only depending on $t$. In this section, we
generalize the $(\bm p,\delta)$-density to $q$-ary sequences and
derive some important properties for these sequences that will be
used in our new construction in the next section.

The $q$-ary $(\bm p,\delta)$-dense sequences can be defined
similar to binary $(\bm p,\delta)$-dense sequences as follows.
\begin{defn}\label{def-density}
Let $d\leq\delta\leq n$ be three positive integers and $\bm
p\in\Sigma_q^d$ called a \emph{pattern}. A sequence $\bm
x\in\Sigma_q^n$ is said to be $(\bm p,\delta)$-\emph{dense} if
each substring of $\bm x$ of length $\delta$ contains at least one
$\bm p$. The indicator vector of $\bm x$ with respect to $\bm p$
is a vector
$$\mathbbm 1_{\bm p}(\bm x)=\big(\mathbbm 1_{\bm p}(\bm
x)_1,\mathbbm 1_{\bm p}(\bm x)_2,\ldots,\mathbbm 1_{\bm p}(\bm
x)_n\big)\in\Sigma_2^n$$ such that for each $i\in[n]$, $\mathbbm
1_{\bm p}(\bm x)_i=1$ if $x_{[i,i+d-1]}=\bm p$, and $\mathbbm
1_{\bm p}(\bm x)_i=0$ otherwise.
\end{defn}

In this work, we will always let $(d=2t)$ $$\bm p=0^t1^t$$ and
view $\bm p=0^t1^t\in\Sigma_q^{2t}$ for any $q\geq 2$. Moreover,
from Definition \ref{def-density}, we have the following simple
remark.
\begin{rem}\label{rem-dense}
Each sequence $\bm x\in\Sigma_q^n$ can be written as the form $\bm
x=\bm x_0\bm p\bm x_1\bm p\bm x_2\bm p\cdots\bm x_{m-1}\bm p\bm
x_{m}$, where each $\bm x_{i}$, $i\in[0,m]$, is a (possibly empty)
string that does not contain $\bm p$. Moreover, $\bm x$ is $(\bm
p,\delta)$-dense if and only if it satisfies: (1) the lengths of
$\bm x_{0}$ and $\bm x_{m}$ are not greater than $\delta-2t$; (2)
the length of each $\bm x_{i}$, $i\in[1,m-1]$, is not greater than
$\delta+1-4t$.
\end{rem}

In \cite{Lenz20}, the VT syndrome of $\text{a}_{\bm p}(\bm x)$ was
used to bound the location of deletions for $(\bm p,
\delta)$-dense $\bm x$, where $\text{a}_{\bm p}(\bm x)$ is a
vector of length $n_{\bm p}(\bm x)+1$ whose $i$-th entry is the
distance between positions of the $i$-th and $(i+1)$-st $1$ in the
string $(1,\mathbbm 1_{\bm p}(\bm x),1)$ and $n_{\bm p}(\bm x)$ is
the number of $1$s in $\mathbbm 1_{\bm p}(\bm x)$. In this paper,
we prove that the VT syndrome of $\mathbbm 1_{\bm p}(\bm x)$ plays
the same role. Specifically, for each $\bm x\in\Sigma_q^n$, let
\begin{align}\label{eq-def-a0}
a_0(\bm x)=\sum_{i=1}^n\mathbbm 1_{\bm p}(\bm x)_i\end{align} and
\begin{align}\label{eq-def-a1}
a_1(\bm x)=\sum_{i=1}^ni\!\cdot\!\mathbbm 1_{\bm p}(\bm
x)_i\end{align} where $\mathbbm 1_{\bm p}(\bm x)$ is the indicator
vector of $\bm x$ with respect to $\bm p$ as defined in Definition
\ref{def-density}. Then we have the following lemma.

\begin{lem}\label{lem-Bnry-burst-Lenz}
Suppose $\bm x\in\Sigma_q^n$ is $(\bm p, \delta)$-dense. For any
$t'\in[t]$ and any $\bm y\in\mathcal B_{t'}(\bm x)$, given
$a_0(\bm x)~(\text{mod}~4)$ and $a_1(\bm x)~(\text{mod}~2n)$, one
can find, in time $O(n)$, an interval $L\subseteq[n]$ of length at
most $3\delta$ such that $\bm y=x_{[n]\backslash D}$ for some
interval $D\subseteq L$ of size $|D|=t'=|\bm x|-|\bm
y|$.\footnote{In fact, we can require that the length of $L$ is at
most $\delta$. However, the proof needs more careful discussions.}
\end{lem}
\begin{proof}
Let $a_0(\bm x)=m$ and $a_0(\bm y)=m'$. Then by Remark
\ref{rem-dense}, $\bm x$ and $\bm y$ can be written as the
following form:
$$\bm x=\bm x_00^t1^t\bm x_10^t1^t\bm x_2\cdots 0^t1^t\bm
x_{m-1}0^t1^t\bm x_m$$ and $$\bm y=\bm y_00^t1^t\bm y_10^t1^t\bm
y_2\cdots 0^t1^t\bm y_{m'-1}0^t1^t\bm y_{m'}$$ where $\bm x_i$ and
$\bm y_j$ do not contain $\bm p=0^t1^t$ for each $i\in[0,m]$ and
$j\in[0,m']$. We denote $$u_i=|\bm y_00^t1^t\bm y_10^t1^t\cdots\bm
y_{i-1}0^t1^t|,~\forall i\in[1,m']$$ and
$$v_i=|\bm y_00^t1^t\bm y_10^t1^t\cdots\bm y_i|,~\forall i\in[0,m'].$$
Additionally, let $u_{0}=0$. Clearly, for each $i\in[0,m']$, we
have $u_i\leq v_i$ and $\bm y_i=y_{[u_i+1,v_i]}$. Moreover, for
each $i\in[0,m'-1]$, each $j_i\in[u_i,v_i]$ and
$j_{i+1}\in[u_{i+1},v_{i+1}]$, we have
\begin{align}\label{dst-dft-intvl}
j_{i+1}-j_i\geq u_{i+1}-v_i\geq 2t.\end{align}

Note that a burst of $t'\leq t$ deletions may destroy at most two
$\bm p$s or create at most one $\bm p$, so $\Delta_0\triangleq
m-m'\in\{-1,0,1,2\}$ and $\Delta_0$ can be computed from $a_0(\bm
x)-a_0(\bm y)$. We need to consider the following four cases
according to $\Delta_0$.

Case 1: $\Delta_0=2$. Then $m'=m-2$ and there is an
$i_{\text{d}}\in[0,m']$ such that $|\bm x_{i_{\text{d}}+1}|\leq
t'-2$ and $\bm y$ can be obtained from $\bm x$ by deleting a
substring $1^{t_1}\bm x_{i_{\text{d}}+1}0^{t_0}$ for some
$t_0,t_1>0$ such that $|\bm x_{i_{\text{d}}+1}|+t_0+t_1=t'$. More
specifically, $\bm y_{i_{\text{d}}}=\bm
x_{i_{\text{d}}}0^t1^{t-t_1}0^{t-t_0}1^t\bm x_{i_{\text{d}}+2}$.
Clearly, we have $2\leq t'\leq t$ and
$x_{[u_{i_{\text{d}}}+1,v_{i_{\text{d}}}+t']}=\bm
x_{i_{\text{d}}}0^t1^t\bm x_{i_{\text{d}}+1}0^t1^t\bm
x_{i_{\text{d}}+2}$. It is sufficient to let
$L=[u_{i_{\text{d}}}+1,v_{i_{\text{d}}}+t']$. But we still need to
find $i_{\text{d}}$.

Consider $\mathbbm 1_{\bm p}(\bm x)$ and $\mathbbm 1_{\bm p}(\bm
y)$. By Definition \ref{def-density}, $\mathbbm 1_{\bm p}(\bm x)$
can be obtained from $\mathbbm 1_{\bm p}(\bm y)$ by $t'$
insertions and two substitutions in the substring $\mathbbm 1_{\bm
p}(\bm y)_{[u_{i_{\text{d}}}+1,v_{i_{\text{d}}}]}$: inserting $t'$
$0$s and substituting two $0$s by two $1$s. Then by
\eqref{eq-def-a1}, we can obtain
\begin{align}\label{eq-2-a1-x}a_1(\bm x)=a_1(\bm
y)+\lambda_1(i_{\text{d}})
+\lambda_2(i_{\text{d}})+(m'-i_{\text{d}})t'\end{align} where
$\lambda_1(i_{\text{d}}),\lambda_2(i_{\text{d}})
\in[u_{i_{\text{d}}}+1,v_{i_{\text{d}}}+t']$ are the locations of
the two substitutions. To find $i_{\text{d}}$, we define a
function $\xi_2$ as follows: For every $i\in[0,m']$, let
$$\xi_2(i)=a_1(\bm y)+2(u_i+1)+(m'-i)t'.$$ Then for each
$i\in[0,m'-1]$, we can obtain
$\xi_2(i+1)-\xi_2(i)=2(u_{i+1}-u_{i})-t'\geq 4t-t'>0$, where the
first inequality comes from \eqref{dst-dft-intvl}. So, for each
$i\in[0,m'-1]$, we have
\begin{align}\label{eq-xi-2-1}
a_1(\bm y)<\xi_2(i)<\xi_2(i+1)\leq\xi_2(m')<a_1(\bm
y)+2n,\end{align} where the last inequality comes from the simple
observation that $\xi_2(m')=a_1(\bm y)+2(u_{m'}+1)<a_1(\bm y)+2n$.

By definition of $\xi_2$ and $a_1$, we can obtain
\begin{align*}\xi_2(i_{\text{d}}+1)-a_1(\bm
x)&=2(u_{i_{\text{d}}+1}+1)-\lambda_1(i_{\text{d}})
-\lambda_2(i_{\text{d}})-t'\\
&\stackrel{(\text{i})}{\geq} 2u_{i_{\text{d}}+1}+2
-2(v_{i_{\text{d}}}+t')-t'\\
&\stackrel{(\text{ii})}{\geq} 4t+2-3t'\\&>0\end{align*} where (i)
holds because $\lambda_1(i_{\text{d}}),\lambda_2(i_{\text{d}})
\in[u_{i_{\text{d}}}+1,v_{i_{\text{d}}}+t']$, and (ii) is obtained
from \eqref{dst-dft-intvl}. On the other hand, by
\eqref{eq-2-a1-x}, $a_1(\bm x)-\xi_2(i_{\text{d}})
=\lambda_1(i_{\text{d}})+\lambda_2(i_{\text{d}})
-2(u_{i_{\text{d}}}+1)\geq 0~($noticing that
$\lambda_1(i_{\text{d}}),\lambda_2(i_{\text{d}})
\in[u_{i_{\text{d}}}+1,v_{i_{\text{d}}}+t'])$. Hence, we can
obtain
\begin{align}\label{eq-xi-2-2}
\xi_2(i_{\text{d}})\leq a_1(\bm
x)<\xi_2(i_{\text{d}}+1).\end{align}

By \eqref{eq-xi-2-1} and \eqref{eq-xi-2-2}, $i_{\text{d}}$ and $L$
can be found as follows: Compute
$$\mu\triangleq a_1(\bm
x)~(\text{mod}~2n)-a_1(\bm y)~(\text{mod}~2n)$$ and
$$\mu_i\triangleq\xi_2(i)~(\text{mod}~2n)-a_1(\bm
y)~(\text{mod}~2n)$$ for $i$ from $0$ to $m'$. Then we can find an
$i_{\text{d}}\in[0,m']$ such that
$\mu_{i_{\text{d}}}\leq\mu<\mu_{i_{\text{d}}+1}$, where
$\mu_{m'+1}=2n$. Let $L=[u_{i_{\text{d}}}+1,v_{i_{\text{d}}}+t']$.
Note that $x_{[u_{i_{\text{d}}}+1,v_{i_{\text{d}}}+t']}=\bm
x_{i_{\text{d}}}0^t1^t\bm x_{i_{\text{d}}+1}0^t1^t\bm
x_{i_{\text{d}}+2}$ and $\bm x$ is $(\bm p,\delta)$-dense, so by
Remark \ref{rem-dense}, the length of $L$ satisfies $|L|=|\bm
x_{i_{\text{d}}}0^t1^t\bm x_{i_{\text{d}}+1}0^t1^t\bm
x_{i_{\text{d}}+2}|\leq 3(\delta+1-4t)+4t\leq 3\delta$, where the
last inequality holds because $2\leq t'\leq t$.

Case 2: $\Delta_0=1$. Then $m'=m-1$ and, similar to Case 1, there
is an $i_{\text{d}}\in[0,m']$ such that $\bm y_{i_{\text{d}}}$ can
be obtained from $\bm x_{i_{\text{d}}}0^t1^t\bm
x_{i_{\text{d}}+1}$ by deleting $t'$ symbols and the pattern
$0^t1^t$ is destroyed. Clearly,
$x_{[u_{i_{\text{d}}}+1,v_{i_{\text{d}}}+t']}=\bm
x_{i_{\text{d}}}0^t1^t\bm x_{i_{\text{d}}+1}$ and it is sufficient
to let $L=[u_{i_{\text{d}}}+1,v_{i_{\text{d}}}+t']$. To find
$i_{\text{d}}$, consider $\mathbbm 1_{\bm p}(\bm y)$ and $\mathbbm
1_{\bm p}(\bm x)$. By Definition \ref{def-density}, $\mathbbm
1_{\bm p}(\bm x)$ can be obtained from $\mathbbm 1_{\bm p}(\bm y)$
by $t'$ insertions and one substitution in the substring $\mathbbm
1_{\bm p}(\bm y)_{[u_{i_{\text{d}}}+1,v_{i_{\text{d}}}]}$:
inserting $t'$ $0$s and substituting a $0$ by a $1$. By
\eqref{eq-def-a1}, we can obtain
\begin{align}\label{eq-1-a1-x}a_1(\bm x)=a_1(\bm
y)+\lambda(i_{\text{d}})+(m'-i_{\text{d}})t'\end{align} where
$\lambda(i_{\text{d}})
\in[u_{i_{\text{d}}}+1,v_{i_{\text{d}}}+t']$ is the location of
the substitution. For every $i\in[0,m']$, let $$\xi_1(i)=a_1(\bm
y)+(u_i+1)+(m'-i)t'.$$ Then for each $i\in[0,m'-1]$, we have
$\xi_1(i+1)-\xi_1(i)=u_{i+1}-u_{i}-t'\geq 2t-t'>0$, and so we can
further obtain
\begin{align}\label{eq-xi-1-1}
a_1(\bm y)<\xi_1(i)<\xi_1(i+1)\leq\xi_1(m')\leq a_1(\bm
y)+n.\end{align}

By definition of $\xi_1$ and $a_1$, we can obtain
$\xi_1(i_{\text{d}}+1)-a_1(\bm
x)=u_{i_{\text{d}}+1}+1-\lambda(i_{\text{d}})-t'
>u_{i_{\text{d}}+1}+1-(v_i+t')-t'\geq 2t+1-2t'>0$.
On the other hand, by \eqref{eq-1-a1-x}, $a_1(\bm
x)-\xi_1(i_{\text{d}})=\lambda(i_{\text{d}})
-(u_{i_{\text{d}}}+1)\geq 0$. Hence, we can obtain
\begin{align}\label{eq-xi-1-2}
\xi_1(i_{\text{d}})\leq a_1(\bm
x)<\xi_1(i_{\text{d}}+1).\end{align}

By \eqref{eq-xi-1-1} and \eqref{eq-xi-1-2}, $L$ can be found as
follows: Compute
$$\mu\triangleq a_1(\bm x)~(\text{mod}~2n)-a_1(\bm
y)~(\text{mod}~2n)$$ and
$$\mu_i\triangleq\xi_1(i)~(\text{mod}~2n)-a_1(\bm y)~(\text{mod}~2n)$$
for $i$ from $0$ to $m'$. Let $i_{\text{d}}\in[0,m']$ be such that
$\mu_{i_{\text{d}}}\leq\mu<\mu_{i_{\text{d}}+1}$. Then let
$L=[u_{i_{\text{d}}}+1,v_{i_{\text{d}}}+t']$, where
$\mu_{m'+1}=2n$. Note that
$x_{[u_{i_{\text{d}}}+1,v_{i_{\text{d}}}+t']}=\bm
x_{i_{\text{d}}}0^t1^t\bm x_{i_{\text{d}}+1}$ and $\bm x$ is $(\bm
p,\delta)$-dense, so by Remark \ref{rem-dense}, $|L|=|\bm
x_{i_{\text{d}}}0^t1^t\bm x_{i_{\text{d}}+1}|\leq
2(\delta+1-4t)+2t<2\delta$.

Case 3: $\Delta_0=0$. Then $m'=m$. For every $i\in[0,m]$, let
$$\xi_{0}(i)=a_1(\bm y)+(m-i)t'.$$ Note that $\bm x$ contains $m$
copies of $0^t1^t$, so we have $n\geq 2tm>mt'$. Therefore, for
each $i\in[0,m-1]$, we can obtain
\begin{align}\label{eq-0-a1-x-0}a_1(\bm y)+n>a_1(\bm
y)+mt'\geq\xi_{0}(i)>\xi_{0}(i+1)\geq a_1(\bm y).\end{align}

As $\Delta_0=0$, there are two ways to obtain $\bm y$ from $\bm
x$:
\begin{itemize}
 \item[1)] There is an $i_{\text{d}}\in[0,m]$ such that
 $\bm y_{i_{\text{d}}}$ can be obtained from $\bm x_{i_{\text{d}}}$
 by a burst of $t'$ deletions. Correspondingly, by Definition
 \ref{def-density}, $\mathbbm 1_{\bm p}(\bm x)$ can be obtained
 from $\mathbbm 1_{\bm p}(\bm y)$ by inserting $t'$ $0$s into
 $\mathbbm 1_{\bm p}(\bm
 y)_{[u_{i_{\text{d}}}+1,v_{i_{\text{d}}}]}$. Therefore, we have
 \begin{align}\label{eq-0-a1-x-1}a_1(\bm x)=a_1(\bm
 y)+(m-i_{\text{d}})t'=\xi_{0}(i_{\text{d}}).\end{align}
 \item[2)] There is an $i_{\text{d}}\in[0,m-1]$ such that
 $\bm x_{i_{\text{d}}}0^t1^t\bm x_{i_{\text{d}}+1}
 =\bm y_{i_{\text{d}}}0^{t+t_0}1^{t+t_1}\bm y_{i_{\text{d}}+1}$
 for some $t_0,t_1\in[1,t'-1]$ such that $t_0+t_1=t'$, and
 $\bm y_{i_{\text{d}}}0^{t}1^{t}\bm y_{i_{\text{d}}+1}$ is obtained
 from $\bm x_{i_{\text{d}}}0^t1^t\bm x_{i_{\text{d}}+1}$ by
 deleting the substring $0^{t_0}1^{t_1}$.
 By Definition \ref{def-density},
 $\mathbbm 1_{\bm p}(\bm x)$ can be obtained from $\mathbbm 1_{\bm
 p}(\bm y)$ by inserting $t_0$ $0$s in $\mathbbm 1_{\bm p}(\bm
 y)_{[u_{i_{\text{d}}}+1,v_{i_{\text{d}}}]}$ and $t_1$ $0$s in
 $\mathbbm 1_{\bm p}(\bm
 y)_{[v_{i_{\text{d}}}+2,v_{i_{\text{d}}+2t}]}$. Therefore, we have
 $$a_1(\bm x)=a_1(\bm y)+t_0+(m-i_{\text{d}}-1)t'.$$ By definition of
 $\xi_{0}$, we
 have $\xi_{0}(i_{\text{d}})-a_1(\bm x)=t'-t_0>0$ and $a_1(\bm
 x)-\xi_{0}(i_{\text{d}}+1)=t_0>0$. So, we can obtain
 \begin{align}\label{eq-0-a1-x-2}
 \xi_{0}(i_{\text{d}})>a_1(\bm
 x)>\xi_{0}(i_{\text{d}}+1)\end{align}
\end{itemize}
For both cases, if $i_{\text{d}}\in[0,m-1]$, then we can
$L=[u_{i_{\text{d}}}+1,v_{i_{\text{d}}}+2t+t']$; if
$i_{\text{d}}=m$, then we can let $L=[u_m+1,n]$. Note that
$x_{[u_{i_{\text{d}}}+1,v_{i_{\text{d}}}+2t+t']}=\bm
x_{i_{\text{d}}}0^t1^t$ and $x_{[u_m+1,n]}=\bm x_{m}$, and since
$\bm x$ is $(\bm p,\delta)$-dense, then by Remark \ref{rem-dense},
we have $|L|=|\bm x_{i_{\text{d}}}0^t1^t|\leq 2\delta$ or
$|L|=|\bm x_{m}|\leq 2\delta$. Moreover, by \eqref{eq-0-a1-x-0},
\eqref{eq-0-a1-x-1} and \eqref{eq-0-a1-x-2}, $i_{\text{d}}~($and
so $L)$ can be found as follows: Compute
$$\mu\triangleq a_1(\bm x)~(\text{mod}~2n)-a_1(\bm
y)~(\text{mod}~2n)$$ and
$$\mu_i\triangleq\xi_0(i)~(\text{mod}~2n)-a_1(\bm y)~(\text{mod}~2n)$$
for $i$ from $0$ to $m$. Then we can always find an
$i_{\text{d}}\in[0,m]$ such that
$\mu_{i_{\text{d}}}\geq\mu>\mu_{i_{\text{d}}+1}$, which is what we
want.

Case 4: $\Delta_0=-1$. Then $m'=m+1$ and there is an
$i_{\text{d}}\in[0,m'-1]$ such that $\bm x_{i_{\text{d}}}=\bm
y_{i_{\text{d}}}0^{t_0}\bm s0^{t-t_0}1^t\bm y_{i_{\text{d}}+1}$ or
$\bm x_{i_{\text{d}}}=\bm y_{i_{\text{d}}}0^t1^{t_1}\bm
s1^{t-t_1}\bm y_{i_{\text{d}}+1}$, where $t_0\in[1,t]$,
$t_1\in[1,t-1]$ and $\bm s\in\Sigma_q^{t'}$, and $\bm y$ can be
obtained from $\bm x$ by deleting $\bm s$. In this case, we can
let $L=[v_{i_{\text{d}}}+1,v_{i_{\text{d}}}+2t+t']$ and can obtain
$|L|=2t+t'<\delta$. To find $i_{\text{d}}$, we consider $\mathbbm
1_{\bm p}(\bm x)$ and $\mathbbm 1_{\bm p}(\bm y)$. By Definition
\ref{def-density}, $\mathbbm 1_{\bm p}(\bm x)$ can be obtained
from $\mathbbm 1_{\bm p}(\bm y)$ by inserting $t'$ $0$s into
$\mathbbm 1_{\bm p}(\bm
y)_{[v_{i_{\text{d}}+1},v_{i_{\text{d}}}+2t]}$ and substituting
$\mathbbm 1_{\bm p}(\bm y)_{v_{i_{\text{d}}}+1}=1$ by a $0$.
Therefore, we have
\begin{align}\label{eq-n1-a1-y}a_1(\bm x)=a_1(\bm
y)-(v_{i_{\text{d}}}+1)+(m'-1-i_{\text{d}})t'.
\end{align}
For every $i\in[0,m'-1]$, let
$$\xi_{-1}(i)=a_1(\bm y)-(v_i+1)+(m'-1-i)t'.$$
Then for each $i\in[0,m'-2]$, we have
$\xi_{-1}(i)-\xi_{-1}(i+1)=v_{i+1}-v_i-t'>0$, where the inequality
is obtained from \eqref{dst-dft-intvl}. Moreover, we have
$\xi_{-1}(0)=a_1(\bm y)-1+(m'-1)t'<a_1(\bm y)+2tm'<a_1(\bm y)+n$
and $\xi_{-1}(m'-1)=a_1(\bm y)-(v_{m'-1}+1)>a_1(\bm y)-n$. So for
each $i\in[0,m'-2]$, we can obtain
\begin{align}\label{eq-n1-a1-x}
a_1(\bm y)+n>\xi_{-1}(i)>\xi_{-1}(i+1)>a_1(\bm y)-n.\end{align} By
\eqref{eq-n1-a1-y} and by the definition of $\xi_{-1}$, we have
$a_1(\bm x)=\xi_{-1}(i_{\text{d}})$. So, by \eqref{eq-n1-a1-x},
$i_{\text{d}}~($and so $L)$ can be found by the following process:
For $i$ from $0$ to $m'-1$, compute $\xi_{-1}(i)$. Then we can
always find an $i_{\text{d}}\in[0,m'-1]$ such that
$\xi_{-1}(i_{\text{d}})~(\text{mod}~2n)=a_1(\bm
x)~(\text{mod}~2n)$, which is what we want.

Thus, one can always find the expected interval $L\subseteq[n]$.
From the above discussions, it is easy to see that the time
complexity for finding such $L$ is $O(n)$.
\end{proof}

\begin{figure*}[htbp]
\begin{center}
\includegraphics[width=18.0cm]{algthm.1}
\end{center}
\label{fig-algorithm-1}
\end{figure*}

\begin{figure*}[htbp]
\begin{center}
\includegraphics[width=18.0cm]{algthm.2}
\end{center}
\label{fig-algorithm-1}
\end{figure*}

In the rest of this section, we will use the so-called sequence
replacement (SR) technique to construct $q$-ary $(\bm
p,\delta)$-dense strings with only one symbol of redundancy for
$\delta=2tq^{2t}\lceil\log n\rceil$. The SR technique, which has
been widely used in the literature $($e.g., see \cite{Wang-S-22},
\cite{Wijngaarden10}-\cite{Sima23}$)$, is an efficient method for
constructing strings with or without some constraints on their
substrings. In this paper, to apply the SR technique to construct
$(\bm p,\delta)$-dense strings, each length-$\delta$ string that
does not contain $\bm p$ needs to be compressed to a shorter
sequence, which can be realized by the following lemma.

\begin{lem}\label{lem-comprs-nd}
Let $\delta=2tq^{2t}\lceil\log n\rceil$ and $\mathcal
S\subseteq\Sigma_q^{\delta}$ be the set of all sequences of length
$\delta$ that do not contain $\bm p=0^t1^t$. For $n\geq
q^{\frac{6t+3-\log_qe}{0.4}}$, there exists an invertible function
$$g:\mathcal S\rightarrow\Sigma_q^{\delta-\lceil\log_qn\rceil-6t-2}$$
such that $g$ and $g^{-1}$ are computable in time $O(\delta)$.
\end{lem}
\begin{proof}
As each $\bm s\in\mathcal S$ has length $\delta=2tq^{2t}\lceil\log
n\rceil$ and does not contain $\bm p$, then $\mathcal S$ can be
viewed as a subset of $\left(\Sigma_q^{2t}\backslash\{\bm
p\}\right)^{q^{2t}\lceil\log n\rceil}$, and we have
\begin{align*}
\log_q|\mathcal
S|&\leq\log_q\left(q^{2t}-1\right)^{q^{2t}\lceil\log n\rceil}\\
&=(2t)q^{2t}\lceil\log n\rceil
+\lceil\log n\rceil\log_q\left(1-\frac{1}{q^{2t}}\right)^{q^{2t}}\\
&\stackrel{(\text{i})}{\leq}(2t)q^{2t}\lceil\log n\rceil+
(\log n+1)\log_q\left(\frac{1}{e}\right)\\
&=\delta-\log_q n\log e-\log_q e\\
&\leq\delta-1.4\log_q n-\log_q e\\
&\stackrel{(\text{ii})}{\leq}\delta-\lceil\log_q n\rceil-6t-2,
\end{align*}
where (i) comes from the fact that
$\left(1-\frac{1}{x}\right)^x<\frac{1}{e}$ for $x\geq 1$, and (ii)
holds when $0.4\log_q n+\log_q e\geq 6t+3$, i.e., $n\geq
q^{\frac{6t+3-\log_qe}{0.4}}$. Thus, each sequence in $\mathcal S$
can be represented by a $q$-ary sequence of length
$\delta-\lceil\log_q n\rceil-6t-2$, which gives an invertible
function $g:\mathcal
S\rightarrow\Sigma_q^{\delta-\lceil\log_qn\rceil-6t-2}$.

Computation of $g$ and $g^{-1}$ involve conversion of integers in
$[0,\left(q^{2t}-1\right)^{q^{2t}\lceil\log n\rceil}-1]$ between
$(q^{2t}-1)$-base representation and $q$-base representation, so
have time complexity $O(2tq^{2t}\lceil\log n\rceil)=O(\delta)$.
\end{proof}

In the rest of this paper, we will always let
$$\delta=2tq^{2t}\lceil\log n\rceil.$$
As we are interested in large $n$, we will always assume that
$n\geq q^{\frac{6t+3-\log_qe}{0.4}}$. The following lemma gives a
function for encoding $q$-ary strings to $(\bm p, \delta)$-dense
strings.

\begin{lem}\label{lem-enc-dss}
There exists an invertible function, denoted by
$\mathsf{EncDen}:\Sigma_q^{n-1}\rightarrow\Sigma_q^{n}$, such that
for every $\bm u\in\Sigma_q^{n-1}$, $\bm x=\mathsf{EncDen}(\bm u)$
is $(\bm p,\delta)$-dense. Both $\mathsf{EncDen}$ and its inverse,
denoted by $\mathsf{DecDen}$, are computable in $O(n\log n)$ time.
\end{lem}
\begin{proof}
Let $g$ be the function constructed in Lemma \ref{lem-comprs-nd}.
The functions $\mathsf{EncDen}$ and $\mathsf{DecDen}$ are
described by Algorithm 1 and Algorithm 2 respectively, where each
integer $i\in[n]$ is also viewed as a $q$-ary string of length
$\lceil\log n\rceil$ which is the $q$-base representation of $i$.

The correctness of Algorithm 1 can be proved as follows:
\begin{itemize}
 \item[1)] In the initialization step, if
 $\tilde{\bm u}=u_{[n-\delta+2t,n-1]}$ contains $\bm p$, then
 clearly, $\bm x$ has length $n$. If
 $\tilde{\bm u}=u_{[n-\delta+2t,n-1]}$ doest not contain $\bm p$,
 then the length of $\bm x$ is $|\bm x|=|(u_{[1,n']},\bm p,\bm p,
 g((\tilde{\bm u},0^{2t})),0^{\lceil\log_qn\rceil+3})|
 =n'+4t+|g((\tilde{\bm u},0^{2t}))|+\lceil\log_qn\rceil+3=n$,
 where $n'=n-\delta+2t-1$ and by Lemma \ref{lem-comprs-nd},
 $|g((\tilde{\bm u},0^{2t}))|=\delta-\lceil\log_qn\rceil-6t-2$.
 So, at the end of
 the initialization step, $\bm x$ has
 length $n$. Moreover, $x_{[n'+1,n'+2t]}=\bm p$ and the
 substring $x_{[n'+2t+1,n]}$ has length $\leq\delta-4t+1$.
 \item[2)] In each round of the replacement step, if
 $\tilde{\bm x}\triangleq x_{[i,i+\delta-1]}$ does
 not contain $\bm p$ for some $i\in[1, n'-\delta+1]$, then by Lemma
 \ref{lem-comprs-nd}, $|(\bm p,\bm p,i,g(\tilde{\bm x}),0,1^{2t},0)|
 =\delta=|x_{[i,i+\delta-1]}|$, so by replacement, the length of
 the appended string equals to the length of the deleted substring,
 and hence the length of $\bm x$ keeps unchanged.
 \item[3)] At the beginning of each round of the
 replacement step, we have $x_{[n'+1,n'+2t]}=\bm p$, so for
 $i\in[n'+2t-\delta+1,n']$, the substring $x_{[i,i+\delta-1]}$
 contains $\bm p$. Equivalently, if
 $\tilde{\bm x}\triangleq x_{[i,i+\delta-1]}$ does
 not contain $\bm p$ for some $i\in[n'-\delta+2,n']$, then it must
 be that $i\in[n'-\delta+2, n'+2t-\delta]$. In this case,
 $|(\bm p,\bm p,i,g((x_{[i,n']},0^{\ell})),0^{},1^{2t-\ell},0)|
 =\delta-\ell=|x_{[i,n']}|$, so by replacement, the length of
 the appended string equals to the length of the deleted substring,
 and hence the length of $\bm x$ keeps unchanged.
 \item[4)] By 1), 2) and 3), the substring
 $x_{[n'+1,n-\delta+1]}$ is always of the form
 $\bm p\bm u\bm p\bm p\bm v\cdots\bm p\bm p\bm w$,
 where all substrings $\bm u, \bm v, \cdots, \bm w$ have length not
 greater than $\delta+1-4t$, so by Remark \ref{rem-dense},
 for each $i\in[n'+1,n-\delta+1]$,
 the substring $x_{[i,i+\delta-1]}$ contains $\bm p$.
 \item[5)] At the end of each round of the replacement step,
 the value of $n'$ strictly decreases, so the \textbf{While}
 loop will end after at most $n$
 rounds, and at this time, for each $i\in[1,n']$,
 the substring $x_{[i,i+\delta-1]}$ contains $\bm p$, which
 combining with 4) implies that $\bm x$ is $(\bm p, \delta)$-dense.
\end{itemize}

The correctness of Algorithm 2 can be easily seen from Algorithm
1, so $\mathsf{DecDen}$ is the inverse of $\mathsf{EncDen}$.

Note that Algorithm 1 and Algorithm 2 have at most $n$ rounds of
replacement and in each round $g~($resp. $g^{-1})$ needs to be
computed, which has time complexity $O(\delta)=O(\log n)$ by Lemma
\ref{lem-comprs-nd}, so the total time complexity of Algorithm 1
and Algorithm 2 is $O(n\log n)$.
\end{proof}

The Algorithm 1 generalizes the Algorithm 2 of \cite{Wang-S-22},
which is for binary sequences. Moreover, our algorithm has only
one symbol of redundancy for all $q\geq 2$, while the algorithm in
\cite{Wang-S-22} has $4t$ bits of redundancy.

\section{Burst-deletion correcting $q$-ary codes}
In this section, using $(\bm p, \delta)$-dense sequences, we
construct a family of $q$-ary codes that can correct a burst of at
most $t$ deletions, where $t$, $q\geq 2$ are fixed integers and
$\delta=2tq^{2t}\lceil\log n\rceil$. In our construction, each
$q$-ary string is also viewed as an integer represented with base
$q$.

Let $\rho=3\delta=6tq^{2t}\lceil\log n\rceil$ and
\begin{equation}\label{def-L-intvl}
L_j=\left\{\!\begin{aligned} &[(j-1)\rho+1,(j+1)\rho],
\!~\text{for}~j\in\{1,\cdots,
\left\lceil n/\rho\right\rceil-2\},\\
&[(j-1)\rho+1, n], ~~~~~~~~~\text{for}~j=\left\lceil
n/\rho\right\rceil-1.
\end{aligned}\right.
\end{equation}
The following remarks are easy to see.
\begin{rem}\label{rem-sets-L}
The intervals $L_j$, $j=1,\cdots,\left\lceil
n/\rho\right\rceil-1$, satisfy:
\begin{itemize}
 \item[1)] For any interval $L\subseteq[n]$ of length at most
 $\rho$, there is a
 $j_0\in\{1,2,\cdots,\left\lceil n/\rho\right\rceil-1\}$
 such that $L\subseteq L_{j_0}$.
 \item[2)] $L_j\cap L_{j'}=\emptyset$ for all $j,j'\in[1,
 \left\lceil n/\rho\right\rceil-1]$ such that $|j-j'|\geq 2$.
\end{itemize}
\end{rem}

The following construction gives a sketch function for correcting
a burst of at most $t$ deletions for $q$-ary sequences.

\textbf{Construction 1}: Let $h$ be the function constructed as in
Lemma \ref{lem-sgbd-a}. Let $L_j$, $j=1,2,\cdots,\left\lceil
n/\rho\right\rceil-1$, be the intervals defined by
\eqref{def-L-intvl}. For each $\bm x\in\Sigma_q^{n}$ and each
$\ell\in\{0,1\}$, let
\begin{align}\label{Con3-h-ell}
\bar{h}^{(\ell)}(\bm x)=\sum_{\substack{j\in[1,
\left\lceil n/\rho\right\rceil-1\}:\\
j\!~\equiv\!~\ell~\text{mod}~2}}
h(x_{L_j})~(\text{mod}~\overline{N}),\end{align} where
$$\overline{N}=q^{4\log_q(2\rho)+o(\log_q(2\rho))}.$$
Then let
\begin{align*}
f(\bm x)=\left(a_0(\bm x)~(\text{mod}~4),a_1(\bm
x)~(\text{mod}~2n),\bar{h}^{(0)}(\bm x), \bar{h}^{(1)}(\bm
x)\right).\end{align*} where $a_0(\bm x)$ and $a_1(\bm x)$ are
defined by \eqref{eq-def-a0} and \eqref{eq-def-a1} respectively.

\begin{thm}\label{thm-burst-del-sketch-1}
For each $\bm x\in\Sigma_q^n$, the function $f(\bm x)$ is
computable in time $O(q^{7t}n(\log n)^3)$, and when viewed as a
binary string, the length $|f(\bm x)|$ of $f(\bm x)$ satisfies
$$|f(\bm x)|\leq\log n+8\log\log
n+o(\log\log n)+\gamma_{q,t},$$ where $\gamma_{q,t}$ is a constant
depending only on $q$ and $t$. Moreover, if $\bm x$ is $(\bm p,
\delta)$-dense, then given $f(\bm x)$ and any $\bm y\in\mathcal
B_{\leq t}(\bm x)$, one can uniquely recover $\bm x$.
\end{thm}
\begin{proof}
By \eqref{eq-def-a0} and \eqref{eq-def-a1}, $a_0(\bm x)$ and
$a_1(\bm x)$ are computable in linear time. By Lemma
\ref{lem-sgbd-a}, the functions $\bar{h}^{(0)}(\bm x)$ and
$\bar{h}^{(1)}(\bm x)$ are computable in time $($noticing that
each $|L_j|=2\rho=6\delta)$
\begin{align*}
O(nq^t|L_j|^3)&=O(nq^t(12tq^{2t}\lceil\log n\rceil)^3)\\
&=O(q^{7t}n(\log n)^3).
\end{align*}
Hence, by Construction 1, we can see that $f(\bm x)$ is computable
in time $O(q^{7t}n(\log n)^3)$.

Since $\delta=2tq^{2t}\lceil\log n\rceil$, then by
\eqref{eq-def-a0}, \eqref{eq-def-a1} and by Lemma
\ref{lem-sgbd-a}, the length of $f(\bm x)$ (viewed as a binary
string) satisfies
\begin{align*}
|f(\bm x)|&=|a_0(\bm x)|+|a_1(\bm x)|+|\bar{h}^{(0)}(\bm
x)|+|\bar{h}^{(1)}(\bm x)|\\&=\log n+3+2\log\overline{N}\\&=\log
n+8\log\rho+o(\log\rho)+\gamma_{q,t}\\&=\log n+8\log\log
n+o(\log\log n)+\gamma_{q,t},\end{align*} where $\gamma_{q,t}$ is
a constant depending only on $q$ and $t$.

Finally, we prove that if $\bm x\in\Sigma_q^n$ is $(\bm
p,\delta)$-dense, then given $f(\bm x)$ and any $\bm y\in\mathcal
B_{\leq t}(\bm x)$, one can uniquely recover $\bm x$.

Suppose $\bm y\in\mathcal B_{t'}(\bm x)$, where $t'=n-|\bm
y|\in[t]$. First, by Lemma \ref{lem-Bnry-burst-Lenz}, from
$a_0(\bm x)~(\text{mod}~4)$ and $a_1(\bm x)~(\text{mod}~2n)$, we
can find an interval $L$ of length at most $\rho=3\delta$ such
that $\bm y=x_{[n]\backslash D}$ for some interval $D\subseteq L$
of size $t'$. By 1) of Remark \ref{rem-sets-L}, there is a
$j_0\in\{1,2,\cdots,\left\lceil n/\rho\right\rceil-1\}$ such that
$L\subseteq L_{j_0}$. Denote $L_{j_0}=[\lambda,\lambda']$. Then we
have: i) $x_{[1,\lambda-1]}=y_{[1,\lambda-1]}$ and
$x_{[\lambda'+1,n]}=y_{[\lambda'+1-t',n-t']}$; ii)
$y_{[\lambda,\lambda'-t']}\in\mathcal
B_{t'}(x_{[\lambda,\lambda']})=\mathcal B_{t'}(x_{L_{j_0}})$. We
can recover $x_{L_{j_0}}$ from $\bar{h}^{(0)}(\bm x)$,
$\bar{h}^{(1)}(\bm x)$ and $y_{[\lambda,\lambda'-t']}$ as follows.

For each $j\in[1,\left\lceil n/\rho\right\rceil-1]$ such that
$j\equiv j_0~(\text{mod}~2)$, by 2) of Remark \ref{rem-sets-L},
$L_j\subseteq[1,\lambda]$ or $L_j\subseteq[\lambda'+1,n]$, so
$h(x_{L_j})$ can be computed from $x_{[1,\lambda-1]}$ and
$x_{[\lambda'+1,n]}$. Moreover, by Lemma \ref{lem-sgbd-a}, we have
$h(x_{L_j})<\overline{N}$. Then $h(x_{L_{j_0}})$ can be solved
from \eqref{Con3-h-ell} and further, by Lemma \ref{lem-sgbd-a},
$x_{L_{j_0}}$ can be recovered from $y_{[\lambda,\lambda'-t']}$.
Thus, $\bm x$ can be recovered from $f(\bm x)$ and $\bm y$, which
completes the proof.
\end{proof}

Now, we can give an encoding function of a family of $q$-ary codes
capable of correcting a burst of at most $t$ deletions.

\begin{thm}\label{thm-bst-enc}
Let
\begin{align*}
\mathcal E: \Sigma_q^{n-1}&\rightarrow ~~~\Sigma_q^{n+r}\\
\bm u~~&\mapsto\big(\bm x, 0^{t}1, f_q(\bm x)\big)\end{align*}
where $\bm x=\mathsf{EncDen}(\bm u)$, $f_q(\bm x)$ is the $q$-ary
representation of $f(\bm x)$ and $r=t+1+|f_q(\bm
x)|=\log_qn+8\log_q\log_qn+o(\log_q\log_qn)+\gamma_{q,t}.$ Then
for each $\bm z=\mathcal E(\bm u)$, given any $\bm y\in\mathcal
B_{\leq t}(\bm z)$, one can recover $\bm x~($and so $\bm z)$
correctly.
\end{thm}
\begin{proof}
Let $t'=|\bm z|-|\bm y|$. Suppose
$D=[i_{\text{d}},i_{\text{d}}+t'-1]\subseteq[1,n+r]$ is an
interval such that $\bm y=z_{[n+r]\backslash D}$. Then we have
$i_{\text{d}}\in[1,n+r-t'+1]$. Clearly, if
$i_{\text{d}}\in[1,n+t+1-t']$, then $y_{n+t+1-t'}=z_{n+t+1}=1$; if
$i_{\text{d}}\in[n+t+2-t',n+r-t'+1]$, then
$y_{n+t+1-t'}=z_{n+t+1-t'}=0$. So, we can consider the following
two cases.

Case 1: $y_{n+t+1-t'}=1$. Then $i_{\text{d}}\in[1,n+t-t'+1]$. We
need further to consider the following three subcases.

Case 1.1: $y_{[n+1-t',n+1+t-t']}=0^{t}1$. In this case, it must be
that $D\subseteq[1,n]$. Therefore, we have
$y_{[1,n-t']}\in\mathcal B_{t'}(\bm x)$ and
$y_{[n+t+2-t',n+r-t']}=f_q(\bm x)$. By Theorem
\ref{thm-burst-del-sketch-1}, $\bm x$ can be recovered from
$y_{[1,n-t']}$ and $y_{[n+t+2-t',n+r-t']}$ correctly.

Case 1.2: There is a $t''\in[1,t'-1]$ such that
$y_{[n+1-t'+t'',n+1+t-t']}=0^{t-t''}1$ and $y_{n-t+t''}\neq 0$. In
this case, it must be that $D=[n+1-t'+t'',n+t'']$. Therefore,
$y_{[1,n+1-t'+t'']}\in\mathcal B_{t'-t''}(\bm x)$ and
$y_{[n+t+2-t',n+r-t']}=f_q(\bm x)$. By Theorem
\ref{thm-burst-del-sketch-1}, $\bm x$ can be recovered from
$y_{[1,n+1-t'+t'']}$ and $y_{[n+t+2-t',n+r-t']}$ correctly.

Case 1.3: $y_{[n+1,n+1+t-t']}=0^{t-t'}1$ and $y_{n}\neq 0$. In
this case, it must be that $D\subseteq[n+1,n+t]$. Therefore,
$y_{[1,n]}=\bm x$.

Case 2: $y_{n+t+1-t'}=0$. Then we have
$i_{\text{d}}\in[n+t+2-t',n+r-t'+1]$ and $\bm x=y_{[1,n]}$.

Thus, $\bm x$ can always be recovered correctly from $\bm y$.
\end{proof}

\section{Conclusions and Discussions}
We proposed a new construction of $q$-ary codes correcting a burst
of at most $t$ deletions. Compared to existing works, which have
redundancy either $\log n+O(\log q\log\log n)$ bits or $\log
n+O(t^2\log\log n)$ bits, our new construction has a lower
redundancy of $\log n+8\log\log n+o(\log\log n)+\gamma_{q,t}$
bits, where $\gamma_{q,t}$ is a constant that only depends on $q$
and $t$.

We can also consider a more general scenario, which allows
decoding with multiple reads (also known as \emph{reconstruction
codes} \cite{KuiCai22}), then with techniques of this work, we can
construct $q$-ary reconstruction codes correcting a burst of at
most $t$ deletions with two reads, and with redundancy $8\log\log
n+o(\log\log n)+\gamma_{q,t}$ bits. This improves the construction
in \cite{Y-Sun-23}, which has redundancy $t(t+1)/2\log\log
n+\gamma_{q,t}'$ bits, where $\gamma_{q,t}'$ is a constant that
only depends on $q$ and $t$. The problem of correcting a burst of
at most $t$ deletions under reconstruction model will be
investigated in our future work.


\end{document}